\documentclass{ifacconf}

\usepackage{graphicx}      
\usepackage{natbib}        
\usepackage{amsfonts} 
\usepackage{algorithm}
\usepackage{algorithmicx}
\usepackage{algpseudocode}
\usepackage{textcomp}
\usepackage{xcolor}
\usepackage{cases}
\usepackage{amssymb}

\usepackage{amsmath}

\usepackage{amsmath}
\usepackage{listings}
\usepackage{bbding}
\usepackage{multirow}
\usepackage{stfloats}
\usepackage{tikz}
\usepackage{algpseudocode}
\usepackage{url}

\newenvironment{proof}{\noindent\textit{Proof.}}{\hfill$\square$\par}

\newtheorem{assumption}{Assumption}
\newtheorem{definition}{Definition}
\newtheorem{lemma}{Lemma}
\newtheorem{theorem}{Theorem}
\newtheorem{problem}{Problem}
\begin{document}
\begin{frontmatter}

\title{Lure-and-Reveal: An Exposure Framework for Stealthy Deception Attack in Multi-sensor Uncertain Systems} 

\thanks[footnoteinfo]{This work is supported by the Guangdong Provincial Project (No.2024QN11X053) and by the Youth S\&T Talent Support Programme of GDSTA (No. SKXRC2025468). (Corresponding author: Bingzhuo Zhong.)}

\author[First]{Meiqi Tian,} 
\author[First]{Yihan Liu,} 
\author[Second]{Bingzhuo Zhong}

\address[First]{The Thrust of Artificial Intelligence, Information Hub, Hong Kong
University of Science and Technology (Guangzhou), Guangzhou
511400, China (e-mail: mtian837@connect.hkust-gz.edu.cn, yliu135@connect.hkust-gz.edu.cn).}
\address[Second]{The Thrust of Intelligent Transportation, System Hub, Hong Kong
University of Science and Technology (Guangzhou), Guangzhou
511400, China (e-mail: bingzhuoz@hkust-gz.edu.cn).}

\begin{abstract}                
Multi-sensor integration via error-state Kalman filter (ES-KF) is widely employed for precise state estimation in cyber-physical systems (CPSs). 
However, this integration exposes the system to stealthy deception attacks that render conventional detection mechanisms ineffective.
We propose an exposure framework to actively reveal such stealthy attacks without modifying sensor interfaces.
The framework introduces a suspect mode in which the defender injects random exposure shakes into the nominal control inputs, thus creating a discrepancy between the defender’s true state estimates and the attacker’s manipulated state estimates, preventing the attack from remaining stealthy.
We further derive an explicit exposure condition that characterizes the minimum shake magnitude to guarantee the finite-time exposure and a compensability condition that ensures the shakes do not degrade closed-loop performance.
Simulation results based on a GNSS/INS-integrated UAV system verify the effectiveness of the proposed framework.
\end{abstract}


\end{frontmatter}

\section{Introduction}
Multi-sensor integration has become an imperative part in modern cyber-physical systems (CPSs), where heterogeneous sensing modalities are fused to generate accurate and consistent state estimates. 
However, the sensor coupling and high reliance on estimator feedback also introduce new vulnerabilities due to the exposure of wireless communication channels and the open operational environments. 
In particular, stealthy deception attacks pose a severe security threat: an intelligent adversary can manipulate one or more sensing streams while remaining undetectable from conventional detectors (\cite{ren2021kullback, geng2024covert}). 
Such attacks may be persistent, enabling the adversary to bias the estimation process over long horizons without raising alarms.

A common line of research relies on modifying the information flow during transmission to expose malicious tampering. 
Watermarking-based approaches perturb the innovation or output signals with encrypted or random patterns and check whether the returned measurements are statistically consistent with the watermark (\cite{wang2021secure, mo2015physical, ahmed2022practical}). In addition, cryptographic solutions, including message authentication codes (\cite{rangwani2021improved}), digital signatures (\cite{prashar2021sdswsn}), and secure timestamping (\cite{zhao2023blockchain}), strengthen measurement integrity by binding each measurement to a cryptographically verifiable identity. Another line of defense that employs active modification of measurement and control packets has been shown to be effective for exposing stealthy false-data injection attacks (\cite{pang2021detection}). 
These methods provide strong detectability guarantees but require direct access to sensor interfaces and I/O communication channels, or additional bandwidth, which are often unavailable in embedded multi-sensor CPSs with fixed sensor firmware and hardware-protected communication channels.

Resilient control frameworks aim to maintain closed-loop stability and good performance under adversarial interference. Examples include resilient CPS architectures (\cite{kim2021stealthy}) and secure model predictive control schemes subject to bounded deception with a known occurrence probability (\cite{wang2020security}). Reinforcement-learning-based secure tracking (\cite{wu2023secure}) and game-theoretic adversarial planning (\cite{athalye2024output}) have also been proposed to mitigate estimation corruption without explicit attack models. However, these strategies typically focus on enhancing robustness rather than revealing the presence of stealthy deception attacks. Consequently, the controller may sacrifice nominal performance to accommodate worst-case disturbances.

In summary, existing defenses either (i) require modifying sensor interfaces or I/O communication channels, (ii) rely on prior assumptions about attack models, or (iii) focus on maintaining robustness without exposing stealthy deception. 
These gaps motivate the exposure framework developed in this paper, and the main contributions are as follows: 
\begin{enumerate}
    \item We present an exposure framework for stealthy deception attacks that operates at the controller level. 
    Random excitations, called \emph{exposure shakes}, are generated to induce a divergence between the defender’s and the attacker’s internal state estimates. This mechanism is compatible with existing conventional detection methods and can be deployed in systems where sensor interfaces are fixed or inaccessible.
    \item We derive two rigorous conditions that characterize the admissible range of the shake magnitude: a finite-time exposure condition that guarantees a stealthy deception attack can be revealed within a given time interval; a compensability condition that ensures the resulting closed-loop deviation remains uniformly bounded. The feasibility interval implied by these two conditions guarantees both finite-time exposure and tolerable state deviation.
\end{enumerate}

The remainder of this paper is organized as follows. Section 2 introduces the problem formulation. 
Section 3 presents the design and theoretical analysis of the stealthy deception attack exposure framework. 
Section 4 shows the simulation results, and Section 5 concludes the paper.

\section{Preliminaries}

\subsection{Notation}
\(\mathbb{N}\) denotes the set of natural numbers. $\mathbb{R}^{n}$ and $\mathbb{R}^{n\times m}$ denote the set of all $n$-dimensional real vectors and $n \times m$ real matrices. Let $I$ denote the identity matrix with an appropriate dimension. $A^\top$ denotes the transpose of matrix \(A\). 
Consider two vectors $a, b \in \mathbb{R}^n$. We write $a \le b$ if each element of $a$ is less than or equal to the corresponding element of $b$.
For any vector \(x \in \mathbb{R}^n\), \(\|x\| = \sqrt{x^\top x}\). \(\|x\|_\infty = \max\limits_{i}|x_i|\) denotes the maximum norm of \(x\). 

\subsection{System Model}
We consider the following discrete-time linear time-invariant systems (dtLTS):
\begin{align}
     x_{k+1} &=Ax_k+Bu_k+w_k, \label{eq: system model 1}
\end{align}
where \(x_k\in \mathbb{R}^{n_x}\) denotes the system state, \(u_k\in \mathbb{R}^{n_u}\) denotes the control input, and $w_k$ denotes the process noise with \(w_k^\top w_k\leq \bar w, \bar w\in \mathbb{R}^+\). 
\(A \in \mathbb{R}^{n_x \times n_x}\) and \(B \in \mathbb{R}^{n_x\times n_u}\) are matrices of compatible dimensions.

In this paper, sensors are categorized into two types based on their susceptibility to attack. \textit{Reliable sensors} are those whose measurements are considered trustworthy and provide accurate short-term estimates. In contrast, \textit{suspicious sensors} are sensing channels that are susceptible to deception attacks while providing valuable measurements for accurate state estimation. 

The drift and accumulated errors of reliable sensor are corrected with the measurements from the suspicious sensor by an error-state Kalman filter (ES-KF), which is expressed as
\begin{equation}\label{eq: estimate}
    \hat x_k = x^n_k + \delta \hat x_k,
\end{equation}
where \(\hat x_k \in \mathbb{R}^{n_x}\) denotes the estimated state, $x^n_k \in \mathbb{R}^{n_x}$ denotes the nominal state provided by reliable sensor, and $\delta \hat x_k \in \mathbb{R}^{n_x}$ denotes the state compensation.

\begin{figure}
    \centering
    \includegraphics[width=1\linewidth]{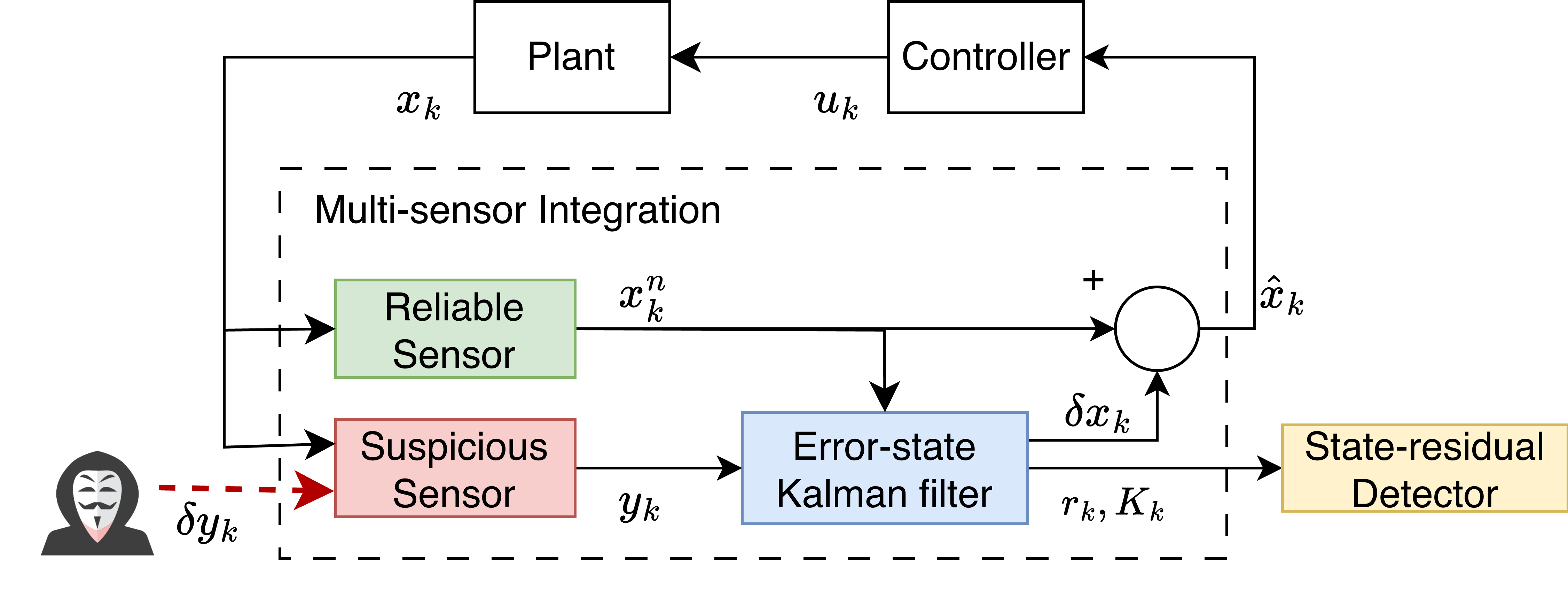}
    \caption{General architecture of CPS under deception attack.}
    \label{fig: system under attack}
\end{figure}

A general structure of the CPS subject to deception attacks is depicted in Fig. \ref{fig: system under attack}. 
Under attack, the corrupted measurement of the suspicious sensor is represented as
\begin{equation}
    y_k = y^{*}_k + \delta y_k,  \label{eq: sensor a output}
\end{equation}
where \(y^{*}_k \in \mathbb{R}^{n_y}\) and $\delta y_k \in \mathbb{R}^{n_y}$ denote the original measurement and the attack vector, respectively.

A state-residual detector is employed in the system to identify malicious biases.
During the ES-KF update, the state residual is
\begin{equation}\label{eq: state-residual}
    r_k = K_{k}(y_k-\hat y_k),
\end{equation}
where \(K_{k} \in \mathbb{R}^{n_x\times n_y}\) denotes the ES-KF gain matrix, 
\(y_k \in \mathbb{R}^{n_y}\) and \(\hat y_k \in \mathbb{R}^{n_y}\) denote the measurement of the suspicious sensor and the predicted measurement from the reliable sensor, respectively.

The state-residual detector is constructed as:
\begin{equation}\label{eq: detector}
    q_k = \max_i \frac{|r_{k,i}|}{\mathbf{T}_i}, 
\end{equation}
where \(q_k \in \mathbb{R}\) denotes the detection statistic at time step \(k\), \(r_k \in \mathbb{R}^{n_x}\) denotes the state residual vector, \(\mathbf{T} \in \mathbb{R}^{n_x}\) denotes the predefined threshold vector, \(r_{k,i}\) and \(\mathbf{T}_i\) denote the \(i\)-th row element of \(r_k\) and of \(\mathbf{T}\), respectively.
An attack is alarmed if \(q_k \geq 1\).

The normal-operation set associated with the state-residual detector \eqref{eq: state-residual} is therefore defined as
\begin{equation}\label{eq: normal-operation set}
\Gamma_k := \{ r_k \in \mathbb{R}^{n_x} : |r_{k,i}| \le \mathbf{T}_i,\ \forall i = 1,\dots,n_x \}.
\end{equation}
\subsection{Problem Formulation}\label{sec: Stealthy Deception Attack}
In this paper, we concern \emph{\(\beta\)-intolerable} attack as defined below.
\begin{definition}\label{def: intolerance attack}
Consider the attack as stated in \eqref{eq: sensor a output}, the detector as in \eqref{eq: detector}, and its normal-operation set $\Gamma_k$ at time step $k\in \mathbb{N}$, i.e., \(\Gamma_k := \{ r_k \in \mathbb{R}^{n_x} : |r_{k,i}| \le \mathbf{T}_i,\ \forall i = 1,\dots,n_x \}.\)
Let $[k_1,k_n]\subset \mathbb{N}$ denote the attack duration, and \(\{\delta y_{1}, \ldots, \delta y_{n}\}\) denote the attack sequence.
Given a prescribed $\beta >0$, the attack is  \emph{\(\beta\)-intolerable} stealthy attack if the induced state residual $r_i \in \mathbb{R}^{n_x}$ satisfies $r_i \in \Gamma_i$ for all $k \in [k_1,k_n]$, and \(\|\delta y_{1}\|_{\infty} \ge \beta\).
\end{definition}




We next specify the attacker’s capabilities and available knowledge in the following assumption.

\begin{assumption}\label{asm:attacker}
We assume that the attacker has perfect knowledge of: (i) the dynamics of dtLTS as in \eqref{eq: system model 1}, including \(A\), \(B\) and \(\bar w\); (ii) the ES-KF gain matrix $K$ as in \eqref{eq: state-residual}; and (iii) the detector threshold vector \(\mathbf{T}\) as in \eqref{eq: detector}.
\end{assumption}
Assumption \ref{asm:attacker} is a commonly used description of an omniscient/white-box attacker (\cite{zhang2024stealthy}).
Under Assumption~\ref{asm:attacker}, the attacker is able to construct its own internal estimation-and-deception process that mirrors the system’s estimates and state evolution, thus remaining stealthy. The attacker-side deception model is therefore given by
\begin{align}\label{eq: attacker dynamics}
     \begin{cases}
    \hat x^{a}_{k+1}= A\hat x_{k}^{a+} + Bu^a_{k}+ w^a_k,\\
    \hat x^{a+}_{k+1}= \hat x^{a}_{k+1}  + \Delta x^a_k,
    \end{cases}
\end{align}
where $\hat x^{a}_{k+1} \in \mathbb{R}^{n_x}$ denotes the estimated state of the target system by the attacker, \(A\) and \(B\) are the same matrices as in \eqref{eq: system model 1}, $w^a_k$ denotes the estimated process noise with ${w_k^a}^\top w_k^a \le \bar w$, $u^a_{k}$ denotes the control input generated by the same control law used by the system, \(\hat x^{a+}_{k+1}\) denotes the deceived state after launching an attack, and \(\Delta x^a_k\) is the induced deviation by the deception attack, which satisfies \(\Delta x^a_k\leq \mathbf{T}\) to maintain stealthiness.

In this paper, our goal is to reveal \(\beta\)-intolerable stealthy attacks within a given time interval.
\begin{definition}\label{def: detector}
    Consider the dtLTS as in \eqref{eq: system model 1}, the \(\beta\)-intolerable stealthy attack as defined in Definition~\ref{def: intolerance attack}, the attacker with capabilities as described in Assumption \ref{asm:attacker}, and the detector as in \eqref{eq: detector}. 
    Let \(q_k \in \mathbb{R}\) denote the detection statistic at time step \(k\), and \(k^a \in \mathbb{N}\) denote the initial occurrence time step of the  \(\beta\)-intolerable stealthy attack.
    Given a time interval \(\mathcal{K}:=[k^a, k^a+K ]\subset\mathbb{N}\), a \((\beta,K)\)-detector is a function
    \begin{equation}
        \Theta: q_k \mapsto \{0,1\},
    \end{equation}
    satisfying the following conditions:
    \begin{itemize}
        \item (C1) For any \(\beta\)-intolerable attack, one has \(\exists k' \in \mathcal{K},\ \Theta(q_{k'}) = 1;\)
        \item (C2) In the absence of an attack, one has \(\forall k \in \mathcal{K},\;\Theta(q_k)=0\).
    \end{itemize}
\end{definition}
Having the notation above, we are ready to state our main problem as follows.
\begin{problem}\label{problem}
    Consider the dtLTS as in \eqref{eq: system model 1}, the \(\beta\)-intolerable stealthy attack as defined in Definition~\ref{def: intolerance attack}, and the attacker with capabilities as described in Assumption \ref{asm:attacker}.
    We aim to construct a \((\beta, K)\)-detector as in Definition~\ref{def: detector} capable of revealing the \(\beta\)-intolerable stealthy deception attack within a maximum of \(K\) time steps after its initial occurrence, without indicating an attack if none has occurred. 
\end{problem}

\section{Stealthy Deception Attack Exposure Framework}
To address Problem \ref{problem}, we propose an exposure framework for stealthy deception attacks which has three operation modes,  as shown in Fig. \ref{fig: framwork}.
\begin{definition}\label{def: mode rule}
Consider the detector as in \eqref{eq: detector}. 
Let  $q_k\in\mathbb{R}$ denote the detection statistic at time step $k$.
Given a suspect threshold $\eta\in(0,1)$, the system mode is determined by the following rule: (i) \emph{Normal mode}: if $q_k < \eta$; (ii) \emph{Suspect mode}: if $\eta \le q_k < 1$; (iii) \emph{Attacked mode}: if $q_k \geq 1$.
\end{definition}

The quantitative relation between the suspect threshold \(\eta\) and tolerance degree \(\beta\) is formulated below.
\begin{lemma} \label{lem:delta-y-lower-bound}
Consider the detector as in \eqref{eq: detector}, and the \(\beta\)-intolerable stealthy attack as defined in Definition~\ref{def: intolerance attack}.
Given the tolerance degree \(\beta \in \mathbb{R}^+\),
if the suspect threshold \(\eta\) employed in Definition \ref{def: mode rule} satisfies
\begin{equation}
    \eta \;\le\;
    \bar\eta(\beta)
    := \frac{\|K_k\|_\infty \,\beta - \bar T}{\bar T},\label{hp}
\end{equation}
where \(K_{k} \in \mathbb{R}^{n_x\times n_y}\) denotes the ES-KF gain matrix at time step \(k\), and $\bar T$ denotes the maximum element of the threshold vector $\mathbf{T}$,
then the \(\beta\)-intolerable stealthy attack as defined in Definition \ref{def: intolerance attack} necessarily leads to \(q_k \;\ge\; \eta\).
\end{lemma}
The proof of Lemma \ref{lem:delta-y-lower-bound} is given in Appendix \ref{appendix: lemma delta-y}.
Lemma \ref{lem:delta-y-lower-bound} guarantees the existence of a suspect threshold such that the initial occurrence of \(\beta\)-intolerable stealthy attack can be identified. 

\begin{figure}
    \centering
    \includegraphics[width=1\linewidth]{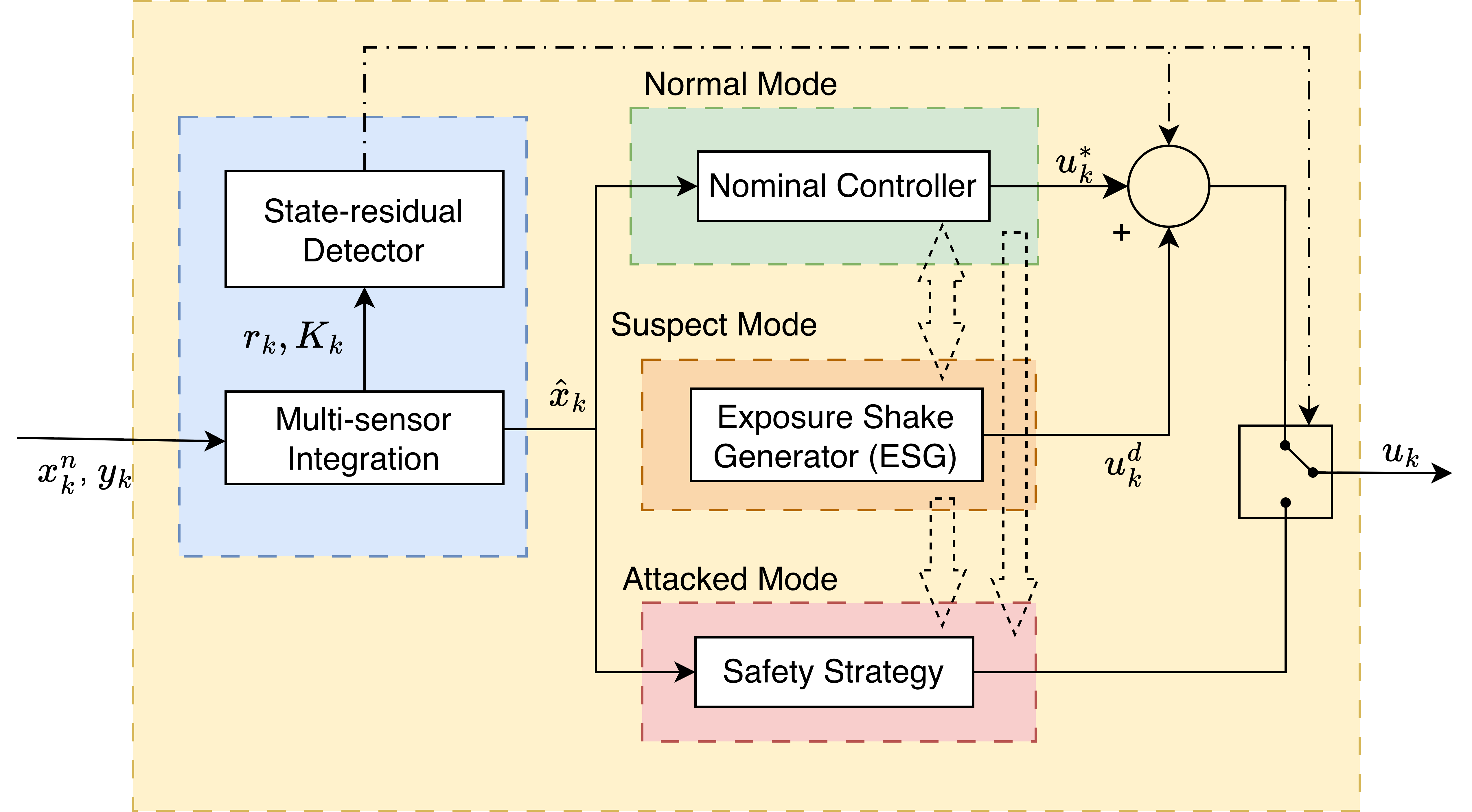}
    \caption{An exposure framework for stealthy deception attack.}
    \label{fig: framwork}
\end{figure}

\textit{Normal Mode $\&$ Attacked Mode} In normal mode, the system continuously performs detection and uses a nominal controller. If the attack is detected in normal mode or exposed in suspect mode, the safety strategy is activated. 
Safe strategies are not the focus of this paper.
In practice, appropriate strategies can be selected based on the sensor type and the operating environment (\cite{wan2020safety, bachrach2011range}).

\textit{Suspect Mode.} This part is our main contribution. 
ESG is designed to generate the exposure shake sequence that guarantees both finite-time exposure and tolerable state deviation.
The exposure shake is injected into the nominal control input to form a composite control input until the attack is revealed.

We assume that the tracking error 
\(
e_k := x_k - x_k^{goal}
\)
evolves according to the following closed-loop dynamics:
\begin{equation}
    e_{k+1} = A_{\mathrm{cl},k} \, e_k + d_k,
    \label{eq:cl-dev}
\end{equation}
where $A_{\mathrm{cl},k} \in \mathbb{R}^{n_x\times n_x}$ denotes the induced state–transition matrix characterizing the error dynamics, and $d_k$ represents external disturbances affecting the error dynamics.
The closed-loop characteristic induced by the nominal controller is discussed below.

\begin{assumption}\label{ass:controller-define}
Consider the dtLTS as in \eqref{eq: system model 1}, and the closed-loop dynamics as in \eqref{eq:cl-dev}.
Define the state–transition matrix from time step $l$ to $k$ as \[\Phi(k,\ell) := 
    A_{\mathrm{cl},k-1}\cdots A_{\mathrm{cl},\ell},k\ge \ell\ge 0.\]
We assume that \eqref{eq:cl-dev} is \emph{uniformly exponentially stable} (UES), 
i.e., there exist constants $c>0$ and $\rho\in (0,1)$ such that \(\|\Phi(k,\ell)\| \le c\,\rho^{\,k-\ell}, \forall\, k\ge \ell\ge 0.\)
\end{assumption}

Assumption~\ref{ass:controller-define} can be satisfied by a range of tracking controllers that ensure exponential convergence of the tracking error (\cite{xiao2017general, miller2020adaptive, chen2019reinforcement}).
In the next subsection, we present the design of exposure shake, followed by the compensation condition.

\subsection{Exposure Shake Generation}
After the system switches to suspect mode at time step \(k^a \in \mathbb{N}\), the exposure shake is injected into the nominal control input to identify the existence of a \(\beta\)-intolerable stealthy attack.
The composite control input is
\begin{equation}\label{eq: control decompose}
    u_k= u^*_k+u^d_k, \; \forall k \in [k^a, k^a+k^{\exp}],
\end{equation}
where $u^*_k \in \mathbb{R}^{n_u}$ denotes the nominal control input provided by the nominal controller, $u^d_k \in \mathbb{R}^{n_u}$ denotes the exposure shake, and \(k^{\exp} \in \mathbb{N}^+\) denotes the exposure horizon.

Intuitively speaking, the randomly generated exposure shakes enlarge the discrepancy between the defender’s true estimates as in \eqref{eq: estimate} and the attacker’s estimates as in \eqref{eq: attacker dynamics}. 
Consequently, the attack conducted based on incorrect estimates can not remain stealthy. 
Fig. \ref{fig: dots} illustrates the attack‐exposure process, and a dynamic demonstration is available online.\footnote{\url{https://github.com/meiqitian01/Lure-and-Reveal}}

\begin{figure}[!t]
    \centering
    \includegraphics[width=1\linewidth]{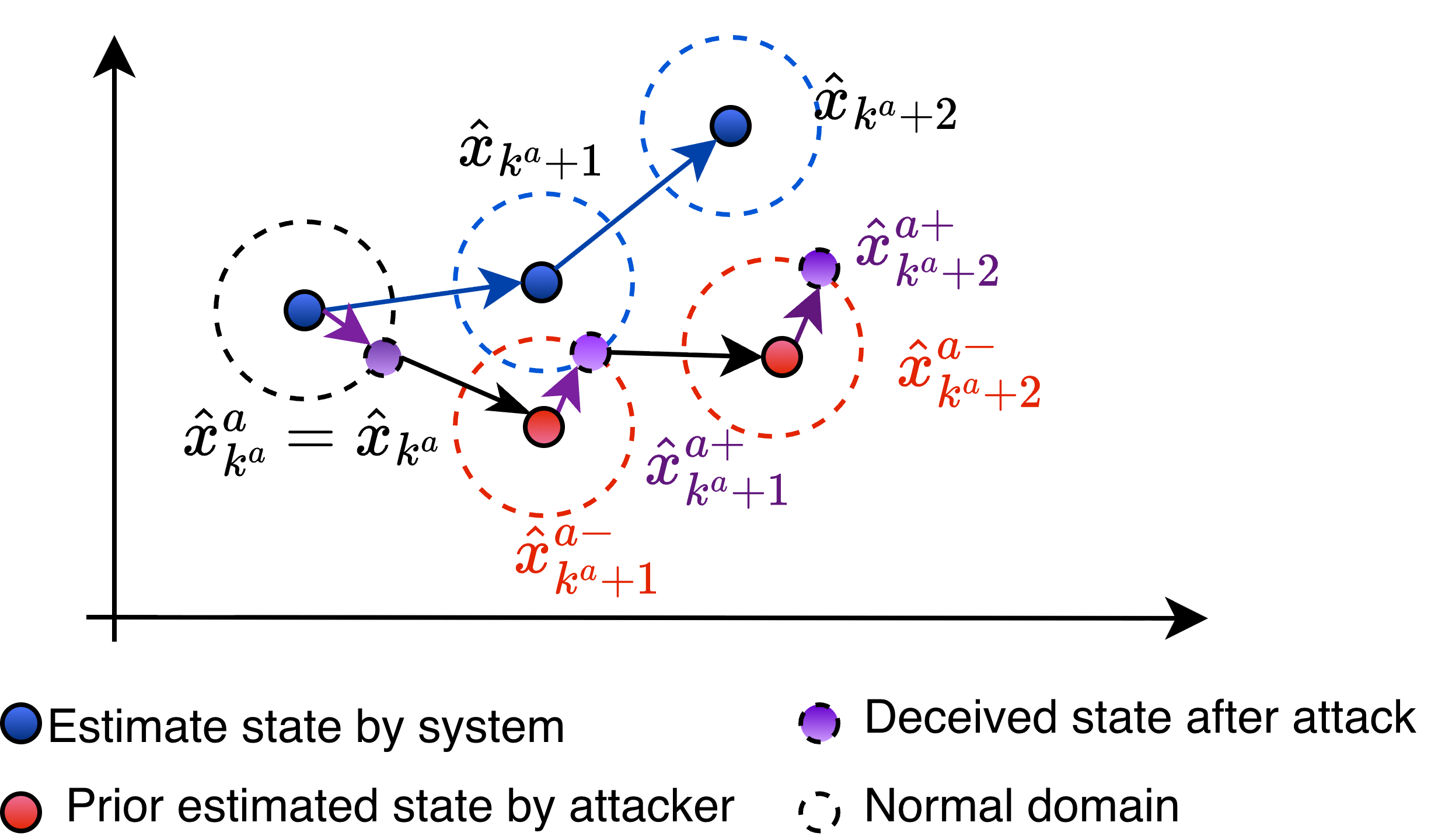}
    \caption{Schematic diagram of the attack exposure process. 
    In the attacker’s anticipation, the estimated states (red dots) evolve along the black arrows and then spoofed to deceived states (purple dots). 
    However, the defender’s estimated states evolve along the blue arrows during exposure process. 
    At time $k^a+2$, any deceived state generated based on $\hat x^a_{k^a+2}$ will exceed the normal domain of $\hat x_{k^a+2}$, triggering attack alarm.}
    \label{fig: dots}
\end{figure}

During the exposure horizon, the system only uses the reliable sensor, i.e., \(\hat x_k = x^n_k\). 
Based on the attacker’s estimate $\hat x^{a}_k$ and the defender’s true estimate $\hat x_k$, an exposure constraint corresponding to the state-residual detector can be expressed as
\begin{equation}\label{eq: exposure constraint}
\|\hat x_{k^a+ k^{\exp}}-\hat x^{a}_{k^a+k^{\exp}}\|_{\infty} -\|2\mathbf{T}\|_{\infty}\geq 0,
\end{equation}
where \(\mathbf{T} \in \mathbb{R}^{n_x}\) is the threshold vector.
This constraint implies that ESG should reveal the stealthy attack by the state-residual detector within an exposure horizon \(k^{\exp}\). 
Intuitively, if the maximum difference between the defender’s true estimate and the attacker’s prior estimate exceeds twice the threshold, then any detection statistic based on the deceived state will exceed the threshold.  

To generate the exposure shake sequence that satisfies the exposure constraint, we propose the following optimization formulation.

\begin{definition}\label{def: opt}
Consider the exposure constraint \eqref{eq: exposure constraint}, and the attacker with capabilities as described in Assumption \ref{asm:attacker}.
Define the following optimization to compute the randomly generated exposure shakes:
    \begin{align}
    \min_{u^d} &\max_{\Delta x^a_{i}\in \Gamma_{i}} J=\sum^{k^a+N}_{i=k^a} ({\tilde x_i}^{\top} Q\tilde x_i + \tilde u_i^\top R \tilde u_i) \label{eq: MPC constraint1}\\
    \text{s.t. } & \hat x^{a+}_{i+1} = A\hat x^{a+}_i + Bu^{a}_i+ w_i^a +{\Delta x^{a}_{i}}\label{eq: MPC constraint2}\\
    & \tilde u_i = u^*_i+u^d_i \label{eq: MPC constraint3}\\
    &\hat x_{i+1} = A \hat x_i + B\tilde u_i+w_i \label{eq: MPC constraint4}\\
    &\|\hat x_{k^a+ k^{\exp}}-\hat x^{a}_{k^a+k^{\exp}}\|_{\infty} -\|2\mathbf{T}\|_{\infty}\geq 0\label{eq: MPC constraint5}\\
    & \|u^{d}_i+\epsilon_i\|\leq \bar u \label{eq: MPC constraint6}
\end{align}
where \(u^*_i \in \mathbb{R}^{n_u}\) and \(u^a_i \in \mathbb{R}^{n_u}\) denote the nominal control inputs of the defender and the attacker generated by the same control law, $N \in \mathbb{N}^+, N > k^{\exp}$ denotes the prediction horizon, $\tilde x_i \in \mathbb{R}^{n_x}$ denotes the difference between the estimated state and the goal state at time step $i$, i.e., $\tilde x_i = x_i-x^{goal}_i$, $Q\in \mathbb{R}^{n_x \times n_x}$ and $R\in \mathbb{R}^{n_u \times n_u}$ denote the symmetric positive definite weight matrices, $\hat x^a_{i} \in \mathbb{R}^{n_x}$ denotes the estimated state by attacker, \(w^a_k \in \mathbb{R}^{n_x}\) denotes the process noise estimated by attacker with \({w^a_k}^\top w^a_k\leq \bar w\), $\epsilon_i \in \mathbb{R}^{n_u}$ denotes a random vector with positive components, and \(\bar u \in \mathbb{R}^+\) denotes the bound of shake magnitude.
\end{definition}

The random vector \(\epsilon_i\) in \eqref{eq: MPC constraint6} is introduced to prevent attackers from inferring the exact value of $u^d$. 
Even if attackers are aware of the suspect mode, they cannot access the true shake values and thus cannot evade exposure.

In the following, the feasibility of the optimization \eqref{eq: MPC constraint1}--\eqref{eq: MPC constraint6} throughout the prediction horizon is established.

\begin{theorem}\label{them: u_min}
Consider the dtLTS as in \eqref{eq: system model 1}, the optimization described in Definition \ref{def: opt}, and the controller that allows Assumptions~\ref{ass:controller-define}to hold. 
Let \(k^a \in \mathbb{N}^+\) denote the time step when the system enters suspect mode. 
Let \(\mathbf{T}=[T_1, T_2, \cdots, T_{n_x}]^\top \in \mathbb{R}^{n_x}\) denote the threshold of the state-residual detector. 
Given a finite exposure horizon \(k^{\exp} \in \mathbb{N}^+\), the optimization \eqref{eq: MPC constraint1}--\eqref{eq: MPC constraint6} has a feasible solution at \(k^a\), if \(\bar u\) in equation \eqref{eq: MPC constraint6} satisfies
\begin{equation}\label{eq: u_min}
        \bar{u} \geq \bar u_{min}:=\frac{1}{\|B\|}(\frac{2\bar T(1-\rho)}{c(1-\rho^{k^{\exp}})}+2\bar w+\bar T),
    \end{equation}
where \(\bar T\) denotes the maximum component of \(\mathbf{T}\), \( c\) and \( \rho\) are the UES constants for \(A_{\mathrm{cl},k}\) in \ref{eq:cl-dev}, \(B\) and \(\bar w\) denote the input matrix and the bound of process noise as in \eqref{eq: system model 1}.
\end{theorem}

Theorem \ref{them: u_min} guarantees the existence of a solution at time step \(k^a\). Next, we will discuss the recursive feasibility of the optimization \eqref{eq: MPC constraint1}--\eqref{eq: MPC constraint6} in \([k^a+1, k^a+N]\).

\begin{lemma}\label{lem:recursive feasibility}
Consider the optimization described in Definition \ref{def: opt}. If there exists a feasible solution at \(k^a\), then there exists a feasible solution for any time step within the prediction horizon \(N\in \mathbb{N}^+, N > k^{\exp}\).
\end{lemma}
The proofs for Theorem \ref{them: u_min} and Lemma \ref{lem:recursive feasibility} are  given in Appendix \ref{appendix: theorem1} and Appendix \ref{appendix: Lemma 1}, respectively. 
These results jointly guarantee the finite-time exposure of stealthy deception attacks.

\subsection{Compensability of exposure shake}\label{sec: The compensability of exposure shake}
To evaluate the impact of the injected shakes on the system, we discuss the compensability of exposure shakes.
\begin{definition}\label{def: compensability}
Consider the optimization defined in Definition \ref{def: opt}.
Let \(\{u^d_{k^a},u^d_{k^a+1}, \cdots, u^d_{k^a+k^{\exp}-1}\}\) denote the exposure shake sequence provided by the optimization. 
This sequence is \emph{\(\varepsilon\)-compensable} if the asymptotic tracking error deviation satisfies \(\limsup_{t \to \infty}\|e_t\| \leq \varepsilon, \varepsilon> 0\).
\end{definition}
Next, we establish a sufficient condition under which the injected shake does not compromise the closed-loop behavior.

\begin{theorem}\label{them:compensability}
Consider the dtLTS as in \eqref{eq: system model 1}, the optimization described in Definition \ref{def: opt}, and the controller that allows Assumptions~\ref{ass:controller-define} to hold. 
Given a prescribed tolerance \(\varepsilon>0\) for the asymptotic tracking error, if the upper bound \( \bar{u} \) of the shake sequence in \eqref{eq: MPC constraint6} satisfies
\begin{equation}\label{eq: compensable condition}
    \bar{u} \leq \frac{(1 - \rho) \varepsilon - c \bar{w}}{c \| B \|},
\end{equation}
then the shake sequence is \emph{\(\varepsilon\)-compensable}.
\end{theorem}
The proof of Theorem~\ref{them:compensability} is given in Appendix \ref{appendix: Theorem 2}.
By Theorem~\ref{them:compensability}, any shake sequence whose magnitude satisfies the compensability condition~\eqref{eq: compensable condition} ensures that the closed-loop deviation from the ideal trajectory remains uniformly bounded by the prescribed tolerance~$\varepsilon$, despite the presence of bounded disturbances.

The requirements of Theorem \ref{them: u_min} and Theorem \ref{them:compensability} are compatible only if the tolerance \(\varepsilon\) is not chosen too small. 
The following lemma characterizes the minimal deviation tolerance that allows a nonempty feasibility interval for the shake magnitude.
\begin{lemma}\label{lem: 1}
Let $\bar u_{\min}$ be the exposure lower bound given by Theorem \ref{them: u_min}, and let
$\bar u_{\max}(\varepsilon)$ be the compensability upper bound given by Theorem \ref{them:compensability}. 
Define
\[
\varepsilon_{\min} := \frac{c(\bar T+3\bar w)}{1-\rho}+\frac{2\bar T}{1-\rho^{k^{\exp}}}.
\]
Then, for any $\varepsilon \ge \varepsilon_{\min}$, \([\bar u_{\min}, \bar u_{\max}(\varepsilon)]\) is nonempty.
\end{lemma}
Lemma \ref{lem: 1} implies that there exists a shake magnitude $\bar u$ that 
both exposes the attack within $k^{\exp}$ steps and keeps the closed-loop  deviation below \(\varepsilon\) with \(\varepsilon \geq \varepsilon_{\min}\). 

Based on the discussion above, we are ready to propose our main results for solving Problem \ref{problem}.
\begin{theorem}
    Consider the dtLTS as in \eqref{eq: system model 1}, the \(\beta\)-intolerable stealthy attack as defined in Definition~\ref{def: intolerance attack}, and the attacker with capabilities as described in Assumption \ref{asm:attacker}.
    One can construct the \((\beta, K)\)-detector as defined in Definition \ref{def: detector}.
\end{theorem}

\section{Simulation}
We conduct simulations on a UAV using tightly-coupled GNSS/INS integration. GNSS is the suspicious sensor, and IMU is the reliable sensor.
\subsection{UAV model and Detectors}
A discrete-time 3D UAV navigation model with state \(x_k=[p_x,p_y,p_z,v_x,v_y,v_z]^\top\) and input \(u_k=[a_x,a_y,a_z]^\top\) is considered, with the sampling period \(dt=0.5\) s.
The system matrices are given by
\[
A = \begin{bmatrix}
    I_3 & dt I_3 \\
    0 & I_3\\
\end{bmatrix}, \quad 
B =  \begin{bmatrix}
\frac{1}{2}dt^2 I_3\\
dt I_3\\
\end{bmatrix}.
\]
The process noise bound $\bar w=0.01$.
A $H_\infty$ controller is designed with $Q = \mathrm{diag}(1,1,1,\allowbreak1,\allowbreak1,1)$, $R = 0.1 I$ and $\gamma = 2.5468$. The exponential stability constants are $c= 3.4157$ and $\rho = 0.5164$. 

The state-residual threshold vector is \(\mathbf{T}=[0.05, \allowbreak 0.05,\allowbreak 0.05,\allowbreak0.05, \allowbreak0.05, \allowbreak0.05]\).
The intolerance degree of attack is \(\beta=0.05\), and the corresponding maximum suspect threshold is \(\bar \eta(\beta)=0.48\).
The suspect threshold is set as \(\eta =0.2\). 
The exposure horizon is \(10\). 
The tolerance of state deviation is \(\varepsilon=0.4\). The bounds of exposure shake are \(\bar{u}_{min} = 0.08418\) and \(\bar{u}_{max} = 0.0905\). 
Additionally, the optimization \eqref{eq: MPC constraint1}--\eqref{eq: MPC constraint6} employs the weighting matrices \(Q = \mathrm{diag}(10,10,10,5,5,5)\) and \(R = \mathrm{diag}(1,1,1)\), and each component of \(\epsilon_i\) lies in \([0, 10^{-3}]\).

\begin{figure*}[t]  
    \centering  
    \includegraphics[width=0.9\textwidth]{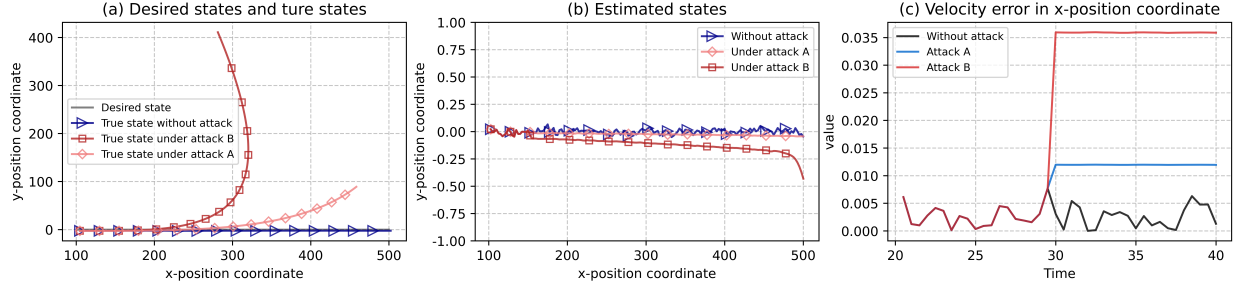}
    \caption{Preset states, true states and estimated states without attack, under Attack A, and under Attack B}  
    \label{fig: normal_attack} 
\end{figure*}

\begin{table*}[!t]  
  \centering
  \caption{Exposure shakes and detection results during the suspect mode}\label{tab: shakes}  
  \begin{tabular}{cccccc}
    \hline
    Attack Group & Shake Step & Nominal Control Input & Exposure Shake & Detection Statistics & Exposure \\
    \hline
    \multirow{3}{*}{Attack A}  
      & $1$ & $[0.1018,-0.0114,-0.1287]^\top$ & $[-0.0004,-0.0834,0.0006]^\top$ & $0.6355$ & -- \\
      & $2$ & $[0.2053,-0.0524,-0.5000]^\top$ & $[-0.0002,-0.0841,-0.0010]^\top$ & $1.7341$ & \Checkmark \\
    \hline
    \multirow{2}{*}{Attack B} 
      & $1$ & $[0.1018,-0.0114,-0.1287]^\top$ & $[0.0000,-0.0834,0.0002]^\top$ & $0.6824$ & -- \\
      & $2$ & $[0.3030,-0.1930,-0.5000]^\top$ & $[-0.0007,-0.0841,-0.0005]^\top$ & $2.9412$ & \Checkmark \\
    \hline
  \end{tabular}
\end{table*}
\subsection{Stealthy Deception Attack}
We consider two stealthy attack intensities, 0.6 and 0.9, where larger intensity corresponds to stronger measurement tampering while remaining undetected.
In both cases, the attack is launched within \(30–100\) s and causes a visible trajectory deviation without triggering alarms.
Fig.~\ref{fig: normal_attack}(a) shows the 2-D plane diagram of the UAV trajectories.
The changes in the estimated states are gentle and undetectable, as shown in Fig.~\ref{fig: normal_attack}(b). 
As for Attack A and Attack B, the maximum discrepancies in the x-coordinate are $1.641$ and $4.791$, respectively, which are far smaller than their impacts on the UAV's true states.

\subsection{Stealthy Attack Exposure}

The increase in the velocity error in x-position coordinate at $30$ s triggers the system to enter suspect mode.
The exposure shakes and detection results calculated until the attack is revealed are listed in Table~\ref{tab: shakes}. 
Revealing Attack A and Attack B both take \(2\) time steps, with the detection statistics being \(1.7341\) and \(2.9412\), respectively.

\section{Conclusion}
 In this paper, an exposure framework to reveal stealthy deception attacks in multi-sensor uncertain systems has been presented. 
The proposed method actively induces detectable inconsistencies in the attacker’s estimates and system's true estimates in a finite time. 
We also provide the formal derivation of a compensability condition that prevents the degradation of control performance.
Simulation results on a UAV navigation system under GNSS attack validated the effectiveness of the proposed method.
Future research directions include extending the proposed exposure framework to nonlinear systems.


\begin{thebibliography}{99}

\bibitem[Ahmed et~al.(2022)Ahmed, Palleti, and Mishra]{ahmed2022practical}
C.M. Ahmed, V.R. Palleti, and V.K. Mishra.
\newblock A practical physical watermarking approach to detect replay attacks in a CPS.
\newblock \emph{Journal of Process Control}, 116:\penalty0 136--146, 2022.

\bibitem[Athalye et~al.(2024)Athalye, Fotiadis, Vamvoudakis, and Hugues]{athalye2024output}
S. Athalye, F. Fotiadis, K.G. Vamvoudakis, and J. Hugues.
\newblock An output feedback game-theoretic approach for defense against stealthy GNSS spoofing attacks.
\newblock In \emph{2024 American Control Conference (ACC)}, pages 3704--3709. IEEE, 2024.

\bibitem[Bachrach et~al.(2011)Bachrach, Prentice, He, and Roy]{bachrach2011range}
A. Bachrach, S. Prentice, R. He, and N. Roy.
\newblock RANGE--robust autonomous navigation in GPS-denied environments.
\newblock \emph{Journal of Field Robotics}, 28(5):\penalty0 644--666, 2011.

\bibitem[Chen et~al.(2019)Chen, Modares, Xie, Lewis, Wan, and Xie]{chen2019reinforcement}
C. Chen, H. Modares, K. Xie, F.L. Lewis, Y. Wan, and S. Xie.
\newblock Reinforcement learning-based adaptive optimal exponential tracking control of linear systems with unknown dynamics.
\newblock \emph{IEEE Transactions on Automatic Control}, 64(11):\penalty0 4423--4438, 2019.

\bibitem[Geng et~al.(2024)Geng, Guo, Tang, Wu, Ren, and Duan]{geng2024covert}
X. Geng, Y. Guo, K. Tang, W. Wu, Y. Ren, and G. Duan.
\newblock A covert spoofing algorithm for SINS/GNSS tightly integrated navigation system.
\newblock \emph{IEEE Transactions on Automation Science and Engineering}, 2024.

\bibitem[Kim et~al.(2021)Kim, Eun, and Park]{kim2021stealthy}
S. Kim, Y. Eun, and K.J. Park.
\newblock Stealthy sensor attack detection and real-time performance recovery for resilient CPS.
\newblock \emph{IEEE Transactions on Industrial Informatics}, 17(11):\penalty0 7412--7422, 2021.

\bibitem[Miller and Shahab(2020)]{miller2020adaptive}
D.E. Miller and M.T. Shahab.
\newblock Adaptive tracking with exponential stability and convolution bounds using vigilant estimation.
\newblock \emph{Mathematics of Control, Signals, and Systems}, 32(3):\penalty0 241--291, 2020.

\bibitem[Mo et~al.(2015)Mo, Weerakkody, and Sinopoli]{mo2015physical}
Y. Mo, S. Weerakkody, and B. Sinopoli.
\newblock Physical authentication of control systems: Designing watermarked control inputs to detect counterfeit sensor outputs.
\newblock \emph{IEEE Control Systems Magazine}, 35(1):\penalty0 93--109, 2015.

\bibitem[Pang et~al.(2021)Pang, Fan, Sun, Liu, and Liu]{pang2021detection}
Z.H. Pang, L.Z. Fan, J. Sun, K. Liu, and G.P. Liu.
\newblock Detection of stealthy false data injection attacks against networked control systems via active data modification.
\newblock \emph{Information Sciences}, 546:\penalty0 192--205, 2021.

\bibitem[Prashar et~al.(2021)Prashar, Rashid, Siddiqui, Kumar, Nagpal, AlGhamdi, and Alshamrani]{prashar2021sdswsn}
D. Prashar, M. Rashid, S.T. Siddiqui, D. Kumar, A. Nagpal, A.S. AlGhamdi, and S.S. Alshamrani.
\newblock SDSWSN---a secure approach for a hop-based localization algorithm using a digital signature in the wireless sensor network.
\newblock \emph{Electronics}, 10(24):\penalty0 3074, 2021.

\bibitem[Rangwani et~al.(2021)Rangwani, Sadhukhan, Ray, Khan, and Dasgupta]{rangwani2021improved}
D. Rangwani, D. Sadhukhan, S. Ray, M.K. Khan, and M. Dasgupta.
\newblock An improved privacy preserving remote user authentication scheme for agricultural wireless sensor network.
\newblock \emph{Transactions on Emerging Telecommunications Technologies}, 32(3):\penalty0 e4218, 2021.

\bibitem[Ren and Yang(2021)]{ren2021kullback}
X.X. Ren and G.H. Yang.
\newblock Kullback--Leibler divergence-based optimal stealthy sensor attack against networked linear quadratic Gaussian systems.
\newblock \emph{IEEE Transactions on Cybernetics}, 52(11):\penalty0 11539--11548, 2021.

\bibitem[Wan et~al.(2020)Wan, Kim, Hovakimyan, Sha, and Voulgaris]{wan2020safety}
W. Wan, H. Kim, N. Hovakimyan, L. Sha, and P.G. Voulgaris.
\newblock A safety constrained control framework for UAVs in GPS denied environment.
\newblock In \emph{2020 59th IEEE Conference on Decision and Control (CDC)}, pages 214--219. IEEE, 2020.

\bibitem[Wang et~al.(2021)Wang, Huang, Wang, and Li]{wang2021secure}
C. Wang, J. Huang, D. Wang, and F. Li.
\newblock A secure strategy for a cyber physical system with multi-sensor under linear deception attack.
\newblock \emph{Journal of the Franklin Institute}, 358(13):\penalty0 6666--6683, 2021.

\bibitem[Wang et~al.(2020)Wang, Ding, and Hu]{wang2020security}
J. Wang, B. Ding, and J. Hu.
\newblock Security control for LPV system with deception attacks via model predictive control: A dynamic output feedback approach.
\newblock \emph{IEEE Transactions on Automatic Control}, 66(2):\penalty0 760--767, 2020.

\bibitem[Wu et~al.(2023)Wu, Yao, Luo, Pan, Sun, Xie, and Wu]{wu2023secure}
C. Wu, W. Yao, W. Luo, W. Pan, G. Sun, H. Xie, and L. Wu.
\newblock A secure robot learning framework for cyber attack scheduling and countermeasure.
\newblock \emph{IEEE Transactions on Robotics}, 39(5):\penalty0 3722--3738, 2023.

\bibitem[Xiao et~al.(2017)Xiao, Dong, Ye, Liu, and Huo]{xiao2017general}
B. Xiao, Q. Dong, D. Ye, L. Liu, and X. Huo.
\newblock A general tracking control framework for uncertain systems with exponential convergence performance.
\newblock \emph{IEEE/ASME Transactions on Mechatronics}, 23(1):\penalty0 111--120, 2017.

\bibitem[Zhang and Li(2024)]{zhang2024stealthy}
D.Y. Zhang and X.J. Li.
\newblock Stealthy attacks against distributed state estimation of stochastic multi-agent systems under composite attack detection mechanisms.
\newblock \emph{Information Sciences}, 672:\penalty0 120584, 2024.

\bibitem[Zhao et~al.(2023)Zhao, Aldyaflah, Gangwani, Joshi, Upadhyay, and Lagos]{zhao2023blockchain}
W. Zhao, I.M. Aldyaflah, P. Gangwani, S. Joshi, H. Upadhyay, and L. Lagos.
\newblock A blockchain-facilitated secure sensing data processing and logging system.
\newblock \emph{IEEE Access}, 11:\penalty0 21712--21728, 2023.







\end{thebibliography}
                                                   

\appendix
\section{Proof of Lemma \ref{lem:delta-y-lower-bound}}\label{appendix: lemma delta-y}
\begin{proof}
With $y_k = y_k^\ast + \delta y_k$, the residual under attack can be
decomposed as
\begin{equation}\label{eq: lemma 1.1}
    r_k = K_k(y_k^\ast - \hat y_k) + K_k\delta y_k
        = r_k^{(0)} + K_k\delta y_k.
\end{equation}

Using the definition of $q_k$, we obtain
\begin{equation}\label{eq: lemma 1.2}
    q_k
    = \max_i \frac{|r_{k,i}|}{T_i}
    \;\ge\;
    \frac{1}{\bar T} \|r_k\|_\infty.
\end{equation}

Employing the triangle inequality on \eqref{eq: lemma 1.1} yields
\begin{equation}\label{eq: lemma 1.3}
    \|r_k\|_\infty
    = \|r_k^{(0)} + K_k\delta y_k\|_\infty
    \ge \|K_k\delta y_k\|_\infty - \|r_k^{(0)}\|_\infty.
\end{equation}

The attack-free bound $q_k^{(0)}\le 1$ implies $\|r_k^{(0)}\|_\infty \le \bar T$.

With \(\|K_k\delta y_k\|_\infty \ge 0\) and \(\|\delta y_k\|_\infty \ge \beta\), we obtain
\begin{equation}
    \|K_k\delta y_k\|_\infty
    \ge \|K_k\|_\infty\,\beta.
\end{equation}

Combining the above inequalities yields
\begin{equation}
    q_k
    \;\ge\;
    \frac{1}{\bar T}
    \big(\|K_k\|_\infty \,\beta - \bar T\big),
\end{equation}
which is the claimed lower bound of the suspect threshold. This proves the lemma.
\end{proof}

\section{Proof of Theorem \ref{them: u_min}}\label{appendix: theorem1}
\begin{proof}
The tracking error dynamics of the defender and attacker are given by
\begin{align}
    e_{k+1} &= A_{\mathrm{cl},k} e_k + B u^d_k + w_k, \label{A.1}\\
    e^{a}_{k+1}
        &= A_{\mathrm{cl},k} e^{a}_k + w^a_k + \Delta x^a_{k+1}. \label{A.2}
\end{align}

Let $\delta e_k := e_k - e_k^{a}$ denote the estimation
difference. Subtracting \eqref{A.2} from \eqref{A.1} yields
\begin{equation}\label{A.3}
    \delta e_{k+1}
    = A_{\mathrm{cl},k}\, \delta e_k
      + B u^d_k
      + \Delta w_k
      - \Delta x^a_{k+1},
\end{equation}
where $\Delta w_k := w_k - w^a_k$ satisfies
$\|\Delta w_k\|_\infty \le 2\bar w$.

Iterating the dynamics from $k^a$ to $k^a+k^{\exp}-1$ gives 
\begin{equation}\label{A.4}
\begin{split}
    &\delta e_{k^a+k^{\exp}}
    = (\Phi(k^a+k^{\exp},k^a)\,\delta e_{k^a}+\sum_{t=0}^{k^{\exp}-1}\\
      &
        \Phi(k^a+k^{\exp},k^a+t+1)
        \big( B u^d_{k^a+t}
            + \Delta w_{k^a+t}
            - \Delta x^a_{k^a+t+1}
        \big).
\end{split}
\end{equation}

Split the contributions into the control term
and the attack term on a specific component:
\begin{equation}\label{A.5}
    S_{\mathrm{ctrl},i} =
    \sum_{t=0}^{k^{\exp}-1}
    (\Phi(k^a+k^{\exp},k^a+t+1)\, B u^{d}_{k^a+t})_i,\\
\end{equation}
\begin{equation}\label{A.6}
    \begin{split}
        S_{\mathrm{att},i} =&
    \sum_{t=0}^{k^{\exp}-1}
    (\Phi(k^a+k^{\exp},k^a+t+1)\cdot\\
    &(\Delta w_{k^a+t}-\Delta x^a_{k^a+t+1}))_i.
    \end{split}
\end{equation}
Using the UES property from Assumption \ref{ass:controller-define}, we obtain
\begin{equation}\label{A.7}
\begin{split}
    |S_{\mathrm{ctrl},i}|
    &\le \sum_{t=0}^{k^{\exp}-1}
            c\rho^{k^{\exp}-t-1}
            \|B\|\,\bar u
    \le
        \frac{c(1-\rho^{k^{\exp}})}{1-\rho}\,
        \|B\|\,\bar u.
\end{split}
\end{equation}

Similarly, using
$\|\Delta w\|\le2\bar w$ and $\|\Delta x^a\|\le \bar T$, one has
\begin{equation}\label{A.8}
    |S_{\mathrm{att},i}|
    \le
        \frac{c(1-\rho^{k^{\exp}})}{1-\rho}\,
        (2\bar w + \bar T),
\end{equation}
where \(\bar T\) is the maximum element of \(\mathbf{T}\).

A sufficient condition for The exposure constraint \eqref{eq: MPC constraint5} is
\begin{equation}\label{A.9}
    |S_{\mathrm{ctrl},i}- S_{\mathrm{att},i}|
    \ge 2\bar T.
\end{equation}

Substituting \eqref{A.7} and \eqref{A.8} into \eqref{A.9} yields
\begin{equation}\label{A.10}
    \frac{c(1-\rho^{k^{\exp}})}{1-\rho}
\left(\|B\|\bar u - (2\bar w + \bar T)\right)
\ge 2\bar T.
\end{equation}

Solving for $\bar u$ gives the lower bound
\begin{equation}\label{A.11}
    \bar u_{\min}
    \ge
    \frac{1}{\|B\|}
    \left(
        \frac{2\bar T(1-\rho)}{c(1-\rho^{k^{\exp}})}
        + 2\bar w + \bar T
    \right).
\end{equation}

Thus, any shake sequence satisfying $\|u^d_k\|\ge \bar u_{\min}$
for $k\in[k^a,k^a+k^{\exp}-1]$ guarantees that the exposure
constraint \eqref{eq: MPC constraint5} is met by time $k^a+k^{\exp}$.
\end{proof}

\section{Proof of Lemma \ref{lem:recursive feasibility}} \label{appendix: Lemma 1}
\begin{proof}
Let \(\{u^{d}_{k^a}, u^{d}_{k^a+1}, \cdots, u^{d}_{k^a+N-1}\}\) denote a feasible solution at \(k^a\). 
We form a solution candidate for \(k^a+1\) as \(\{u^{d}_{k^a+1}, \cdots, u^{d}_{k^a+N-1}, u^d\}\) with \(\|u^d\|\leq \bar u\). 
The original feasible sequence is constructed to guarantee the exposure within time \(k^a+k^{exp}\), thus \(u^{d}_{k^a+N-1}\) ensures exposure within \(k^a+k^{exp}\). By repeatedly shifting the sequence forward by one step, the feasibility is preserved at every time step prior to exposure.
\end{proof}

\section{Proof of Theorem \ref{them:compensability}}\label{appendix: Theorem 2}
\begin{proof}

Iterating the deviation dynamics \ref{eq:cl-dev} for $\ell$ steps gives
\begin{equation}\label{C.2}
e_{k+\ell}
    = \Phi(k+\ell,k) e_k
        + \sum_{i=0}^{\ell-1}
            \Phi(k+\ell,k+i+1)\, d_{k+i}.
\end{equation}

Taking norms and using the UES property, we obtain
\begin{equation}\label{C.4}
    \|e_{k+\ell}\|
        \le \rho^{\ell}\|e_k\| +\frac{c}{1-\rho}\|d_{k+i}\|.
\end{equation}

The bound of the composite disturbance \(d_k=Bu^d_k + w_k\) is 
\begin{equation}\label{C.5}
    \bar d=\|B\|\,\bar u + \bar w.
\end{equation}

Taking $\limsup_{\ell\to\infty}$ and noting
$\rho^\ell\to 0$, we obtain the steady-state deviation bound:
\begin{equation}\label{C.6}
    \limsup_{t\to\infty} \|e_t\|
    \le
    \frac{c}{1-\rho}\,\bar d.
\end{equation}

Substituting \eqref{C.5} into \eqref{C.6}, the condition $\limsup_{t\to\infty}\|e_t\|\le \varepsilon$ is guaranteed if
\begin{equation}\label{C.7}
    \frac{c}{1-\rho}(\|B\|\bar u + \bar w) \le \varepsilon.
\end{equation}

Solving for $\bar u$ yields the compensability condition
\begin{equation}\label{C.8}
    \bar u
    \le
    \frac{(1-\rho)\varepsilon - c\bar w}{c\,\|B\|},
\end{equation}

which proves Theorem \ref{them:compensability}.
\end{proof}

\end{document}